\documentclass[journal]{IEEEtran}
\makeatletter
\def\ps@headings{%
\def\@oddhead{\mbox{}\scriptsize\rightmark \hfil \thepage}%
\def\@evenhead{\scriptsize\thepage \hfil \leftmark\mbox{}}%
\def\@oddfoot{}%
\def\@evenfoot{}}
\makeatother 

\usepackage[final]{graphicx}
\usepackage[reqno]{amsmath}
\usepackage{amssymb}
\usepackage{cite}
\usepackage{balance}
\usepackage{color}
\usepackage{caption}

\usepackage[ruled,vlined]{algorithm2e}

\usepackage{floatrow}

\newfont{\bbb}{msbm10 scaled 500}

\newfont{\bb}{msbm10 scaled 1100}

\newcommand{\argmax}{\operatornamewithlimits{argmax}}

\newtheorem{theorem}{Theorem}

\newtheorem{lemma}{Lemma}

\IEEEoverridecommandlockouts

\author
{Mehmet Karaca}
\title{Scheduling and Dynamic Pilot Allocation For Massive MIMO with Varying Traffic
\thanks{
Mehmet Karaca is with the Department of Electrical-Electronics Engineering, TED University, 06420-Cankaya,
Ankara/Turkey. Email: mehmet.karaca@tedu.edu.tr,
}
}
\begin{document}
\maketitle

%%%%%%%%%%%%%%%%%%%%%%%%%%%%%%%%%%%%%%%%%%%%%%%%%%%%%%%%%%%%%%%%%%%%%%%%%%%%%%%%%%%%%%%%%%%%%%%%%%%%%%%%%%%%%%
%%%%%%%%%%%%%%%%%%%%%%%%%%%%%%%%%%%%%%%%%%%%%%%%%%%%%%%%%%%%%%%%%%%%%%%%%%%%%%%%%%%%%%%%%%%%%%%%%%%%%%%%%%%%%%

\begin{abstract}
In this paper, we consider the problem of joint user scheduling and dynamic pilot allocation in a Time-Division Duplex (TDD) based  Massive MIMO network under varying traffic condition.  One of the main problems with Massive MIMO systems is that the number of available orthogonal pilot signals is limited, and the dynamic allocation of these signals to different users is crucially needed to utilize the full benefit of these systems. In addition, pilot signals are radio resource control (RRC) configured in practice, and hence the frequent reconfiguration  causes high signaling overhead and is costly. Using  Lyapunov optimization framework, we develop an optimal algorithm that first  assigns pilots  dynamically based on  queue sizes and the channel conditions of users as well as the reconfiguration cost at large time-scale. Then, it  schedules users on a small-time scale.  Numerical results show the efficacy of our algorithm and demonstrate that pilots do not need to be configured frequently at the expense of increased queue delay.  
\end{abstract}

%%%%%%%%%%%%%%%%%%%%%%%%%%%%%%%%%%%%%%%%%%%%%%%%%%%%%%%%%%%%%%%%%%%%%%%%%%%%%%%%%%%%%%%%%%%%%%%%%%%%%%%%%%%%%%%%%%
%%%%%%%%%%%%%%%%%%%%%%%%%%%%%%%%%%%%%%%%%%%%%%%%%%%%%%%%%%%%%%%%%%%%%%%%%%%%%%%%%%%%%%%%%%%%%%%%%%%%%%%%%%%%%%%%%%

\section{Introduction}
\label{sec:intro} 
\IEEEPARstart{T}{he} Massive MIMO is one of the core technologies in future 5G  wireless  systems in which base stations (BSs) operate with a large number  of antennas. Using many antennas at the BS provides very high beamforming gain and the capability of serving multiple users simultaneously via spacial multiplexing at the same time and frequency resources. Consequently, enormous spectral efficiency can be achieved~\cite{TMbook},~\cite{TMMag}.

One of the main challenges  for a TDD based Massive MIMO system is  the accurate estimation of the channel of users. The estimation is realized through special uplink signals called as \textit{pilot} signal or \textit{Sounding Reference Signal} (SRS)\footnote{ In the rest of the paper, we use these terms interchangeably.}  allocated to users. Multiplexing and   beamforming gain can be achieved  only by the users with pilot signals, and the gain  and the system performance degrade when those signals are not carefully assigned to the users.  For instance,  the users with high amount of traffic may be preferred to have pilot signals to maximize the network throughput. On the other hand, the practical problems that obstruct the pilot allocation must also be taken into account. Specifically, pilot signals are configured through RRC signaling~\cite{book} and the frequent reconfiguration of the pilots can cause an intolerable signaling overhead  in practice. Therefore, in addition to a smart SRS allocation, the minimization of the frequency of SRS configuration needs  attention.

Although the  impact of reusing the same pilot signals on different cells (i.e., pilot contamination) has been well investigated in the literature, the allocation of these signals with the consideration of varying user's traffic and channel conditions  has received little attention. In~\cite{7878690},~\cite{7852217},~\cite{inproceedings} and~\cite{8377316} the pilot allocation and scheduling are considered without the impact of the network traffic. In~\cite{8949102}, the pilot allocation is done for only special messages and the results cannot be generalized.

Our contributions are summarized as follows: i-) We formulate the problem of the scheduling and pilot allocation for Massive MIMO networks with the associated signaling cost as a stochastic optimization problem   under Lyapunov optimization framework; ii-) We develop an optimal algorithm that operates at two different time-scale and does not need the future knowledge of the system. We also derive analytical bounds on the performance of the optimal algorithm in terms of average queue size and the signaling cost; iii-) We implement a realistic network setting,  and demonstrate  the performance of our algorithm and depict the tradeoff between the signaling cost and the average queue size.

%%%%%%%%%%%%%%%%%%%%%%%%%%%%%%%%%%%%%%%%%%%%%%%%%%%%%%%%%%%%%%%%%%%%%%%%%%%%%%5%%%%%%%%%%%%%%%%%%%%%%%%%%%%%%%%%%%%%%%%
%%%%%%%%%%%%%%%%%%%%%%%%%%%%%%%%%%%%%%%%%%%%%%%%%%%%%%%%%%%%%%%%%%%%%%%%%%%%%%%%%%%%%%%%%%%%%%%%%%%%%%%%%%%%%%%%%%%%%%%
\section{System Model and Problem Formulation}
\label{sec:model} We consider a Massive MIMO capable cellular network where there is a BS with $M $ antennas serving  $N$ users. Let $ \mathcal{N}$ be the set of users with single antenna. The BS can transmit simultaneously  up to $K$ users via its Multi-User (MU)-MIMO capability and $K< M$. The system schedules users on time-slot fashion at each $\tau$ (e.g., regular scheduling decision that LTE performs at every 1 ms) and let $k(\tau)$ be the number of users  served simultaneously at time $\tau$, $1 \leq k(\tau) \leq K$.  In this case, the transmit power for a scheduled user $n$ is $p_n(\tau) = \frac{P_{tot}}{k(\tau)}$, and $P_{tot}$ is the total transmit power. 
%Also, let $\textbf{H}(\tau) \in \mathbb{C}^{NxM}$ the channel matrix between the $M$ antennas BS and the $l(t)$ users at time slot $\tau$. 
The transmission rate for a user $n$ at time slot $\tau$, $R_n(\tau) \leq R_{max}$ for all $n$ and $\tau$, is given by,
%\begin{align}
%y_n(\tau) = s_n(\tau)&\sqrt{p_n(\tau)}\textbf{h}_n^\dagger(\tau)\textbf{w}_n(\tau) \\ \notag + \sum_{\substack{k=1 %\\ k\neq n}}^{k(t)} &s_k(\tau)\sqrt{p_k(\tau)}\textbf{h}_n^\dagger(\tau)\textbf{w}_k(\tau) + \sigma_n^2(\tau)
%\end{align}
\begin{align*}
R_n(\tau)=\textrm{B}\gamma\log_2\left(1 +  \frac{p_n(\tau)\mid\textbf{h}_n^\dagger(\tau)\textbf{w}_n(\tau)\mid^2  }{  \sum_{\substack{k=1 \\ k\neq n}}^{k(\tau)} p_k(\tau)\mid\textbf{h}_n^\dagger(\tau)\textbf{w}_k(\tau)\mid^2 + \sigma^2 }   \right), 
\end{align*}
where $\textbf{h}_n^\dagger(\tau) \in \mathbb{C}^{M}$ is channel vector of user $n$ in the downlink direction, and $s_n(\tau)$ represents the complex transmit symbol for user $n$. Also, $\textbf{w}_n(\tau) \in \mathbb{C}^{N}$ is the normalized precoding vector of  user $n$ and $\sigma_n^2(\tau)$ is zero mean complex Gaussian additive noise with power $\sigma^2$. Furthermore, $\textrm{B}$ is the system bandwidth and $\gamma$ is the fraction of the total time/frequency resource used for data transmission, and the $(1-\gamma)$ fraction is donated for obtaining SRS.  We note that $\gamma$ depends on the length of the coherence block and pilot signals~\cite{TMbook}.

Let $I_n(\tau)$ be the scheduling decision given for user $n \in \mathcal{N}$ at time slot $\tau$. If user $n$ is scheduled then $I_n(\tau)=1$, else $I_n(\tau)=0$.  We assume there are $P$ number of pilot signals that can be allocated  among $N$ users,  where $K < P < N$.  The decision for SRS allocation  is taken at a larger time-scale denoted as $T > \tau$.   In every time slot of the form $t =lT$ where $l = 1,2,....$, the decision denoted as $I_n^s(t)$ is taken to decide whether  user $n$,  $n \in \mathcal{N}$, should have SRS or not. if $I_n^s(t)=1$, the system allocates  SRS to user $n$ and it can benefit from the multiplexing and beam forming gain between  $ [t,t+T-1]$. If $I_n^s(t)=0$, user $n$ cannot have  a SRS till the next time $t=(k+1)T$. After deciding  on $I_n^s(t)$ for every user,  a SRS configuration flag is set. If $I_n^s(t) = I_n^s((k-1)T)$ for all users (i.e. the same SRS allocation as previous time), then the decision on the  reconfiguration of SRS denoted as  $I^{s}(t)$ is set to 0,  otherwise $I^{s}(t)=1$.  That is to say  when  $I_n^s(t)$ is given for all users,  $I^s(t)$ is completely determined.

At each time slot $\tau$,  data randomly arrives to the queue of each users. Let $A_n(\tau)$ be the amount of data (bits or packets) arriving into
the queue of user $n$ at time slot $\tau$. We assume that $A_n(\tau)$ is a
stationary process and it is independent across users and time
slots, and $A_n(\tau) \leq A_{max}$ for all $n$ and $\tau$. We denote the arrival rate vector as $\boldsymbol
\lambda=(\lambda_1,\lambda_2,\cdots,\lambda_N)$, where $\lambda_n =
{\mathbb E}[A_n(\tau)]$. Let $\boldsymbol
Q(\tau)=(Q_1(\tau),Q_2(\tau),\cdots,Q_N(\tau))$ denote the vector of queue
sizes, where $Q_n(\tau)$ is  the queue length of user $n$ at time slot
$\tau$. A queue is strongly stable if $\limsup_{t\rightarrow \infty}\frac{1}{t} \sum_{\tau=0}^{t-1}\mathbb{E}(Q_n(\tau)) < \infty $. Moreover, if every queue in the network is stable then the network
is called stable. The dynamics of the queue of user $n$ is
\begin{equation}
Q_n(\tau+1)=\max[Q_n(\tau)- I_n(\tau)R_{n}(\tau) , 0] + A_n(\tau).  \label{eq:queuelength}
\end{equation}
Let $\Lambda$ denote the capacity region of the system, which is the largest possible set of rates $\boldsymbol \lambda$ that can be supported by a joint scheduling and SRS allocation algorithm with ensuring the network stability.

%Recall that configuring SRS allocation frequently (i.e.,  small T values) causes huge signaling overhead and is not affordable.  Hence, the minimization of the frequency of SRS configuration must be aimed. However, the SRS configuration should also be sufficiently adaptive and dynamic so that the users with the desired properties, such as with good channel condition or with high or specific traffic type, can have more opportunity to access to the channel with higher SINR gain to enhance the performance in terms of throughput, latency or reliability.%  

We recall that configuring SRS allocation frequently (i.e.,  at every $T$) causes  signaling overhead and costly but  the SRS allocation should be also sufficiently adaptive and dynamic. Let $C$ be the cost when SRS allocation is reconfigured, i.e., $I^s(t)=1$.  Then, the average cost is given $C_{avg} = \limsup_{t\rightarrow \infty}\frac{1}{t}
\sum_{l=0}^{t-1} \mathbb{E} [C I^s(l)] $. 
%\begin{equation*}
%C_{avg} = \limsup_{t\rightarrow \infty}\frac{1}{t}
%\sum_{l=0}^{t-1} \mathbb{E} [C I^s(l)] \label{eq:defination}
%\end{equation*}
The control decisions of the system are $c_n(t,\tau) = [I_n^s(t), I_n(\tau)]$ for user $n$. Let $ \mathcal{C}(t,\tau)$ be the set of all possible
control decisions. We consider the following optimization problem:
\begin{align}
&\min \quad C_{avg} \\
\text{s.t.:1)}\  &\text{Netwrok stability}\\
\text{2)}\ & c_n(t,\tau) \in \mathcal{C}(t,\tau), \quad \forall n, t, \tau
\label{eq:problem}
\end{align}
The problem (2)-(3)-(4) aims at minimizing the average cost by taking the scheduling and SRS allocation decisions optimally. The problem constitutes a stochastic optimization
problem and we next propose a solution based on Lyapunov
optimization technique.
\section{Joint SRS Configuration, Allocation and Scheduling}
In our work, we use Lyapunov drift and optimization tools
~\cite{Georgiadis:Resource06}. The advantage
of this tool is the ability to deal with performance optimization
and queue stability problems simultaneously in a unified framework. We first define
quadratic Lyapunov function as $L(\mathbf{Q}(t))\triangleq \frac{1}{2} \sum_{n=1}^N(Q_n^2(t))$ measuring
the total queue size in the system. We then define the conditional T-slot Lyapunov drift that is the expected variation in the Lyapunov function over $T$ slots as follows:
\begin{equation}
\Delta_T(t)\triangleq \mathbb{E}
\left[L(\mathbf{Q}(t+T))-L(\mathbf{Q}(t))|\mathbf{Q}(t) \right].  \label{eq:drift}
\end{equation}
The  Lyapunov optimization tool allows us to minimize the drift and  optimize a given objective simultaneously~\cite{Georgiadis:Resource06}.  We add our system cost $V \mathbb{E} [I^s(t) | (\mathbf{Q}]$ (i.e., $T$ slot drift-plus-penalty ) to \eqref{eq:drift}.
\begin{equation}
\Delta_T(t) + V \mathbb{E} [I^s(t) | (\mathbf{Q(t)}]  \label{eq:driftpluspen},
\end{equation}
where $V$ is system parameter that characterizes a tradeoff between performance
optimization and delay in the queues. According to the Lyapunov optimization theory, the problem (2)-(3)-(4) can be reinterpreted as the minimization of \eqref{eq:driftpluspen} which can be done by  first deriving an upper bound  for \eqref{eq:driftpluspen} in the following Lemma.
\begin{lemma}
\label{lemma:1}
Given $V >1$, and at time $t=lT$,  for any feasible decision, we have
\begin{align}
\Delta_T(t) & + V \mathbb{E} \left[ I^s(t) | (\mathbf{Q(t)}\right]  \nonumber\\
&\leq   B_1  +\mathbb{E} \left[\sum_{\tau = t}^{t +T-1} \sum_{n=1}^{N} Q_n(\tau)A_n(\tau)   | \mathbf{Q(t)}\right] \nonumber\\
& - \mathbb{E} \left[\sum_{\tau = t}^{t +T-1} \sum_{n=1}^{N} Q_n(\tau) R_n(\tau) -CVI^s(t)   | \mathbf{Q(t)} \right]\label{eq:bound2}
\end{align}
where $B_1= \frac{ NT (R_{max}^2 + A_{max}^2)}{2}$.
\end{lemma}
\begin{proof}
The proof is given in Appendix \ref{sec:lemma1}.
\end{proof}

Now, our aim  is to find a method that minimizes  the right hand side of \eqref{eq:bound2}, and this is realized by maximizing the following term in (8) given in the following problem.

\textbf{Opt 1:} Maximize over $c_n(t,\tau) \in \mathcal{C}(t,\tau),  \forall n, t, \tau$:
\begin{align}
\mathbb{E} \left[\sum_{\tau = t}^{t +T-1} \sum_{n=1}^{N} Q_n(\tau) [R_n(\tau)-A_n(\tau)] -CVI^s(t)  | \mathbf{Q(t)}  \right]
\end{align}
%The optimization is performed over two control decisions: i-)  it decides on $I_n^s(t)$ for all users and determines  $I^s(t)$ ii-) After that, at each time $\tau$, it selects users for scheduling.

It is clear that the solution of Opt 1 requires the prior knowledge of the future queue sizes and the data rates which depends on the future channel conditions over $[t, t+T-1]$,  and this knowledge cannot be obtained in practice. In order to overcome this issue, first we follow the idea in~\cite{6810886} and approximate  the future queue sizes as the current observation, i.e., $Q_n(\tau) = Q_n(t)$ for all $\tau \in [t, t+T-1]$ and $n \in \mathcal{N}$. Then, we obtain a looser but more relaxed bound as follows.
\begin{lemma}
	\label{lemma:2}
	Given $V >1$, and at time $t=lT$,  for any feasible decision, we have the following bound,
	\begin{align}
	\Delta_T(t) & + V \mathbb{E} \left[ I^s(t) | (\mathbf{Q(t)}\right]  \nonumber\\
	&\leq   B_2  + \mathbb{E} \left[\sum_{\tau = t}^{t +T-1} \sum_{n=1}^{N} Q_n(t)A_n(\tau) | \mathbf{Q(t)} \right] \nonumber\\
	& - \underbrace{\mathbb{E} \left[\sum_{\tau = t}^{t +T-1} \sum_{n=1}^{N} Q_n(t) R_n(\tau) -CVI^s(t)   | \mathbf{Q(t)} \right]}_{\textit{Third  Term}}\label{eq:bound3}
	\end{align}
	where $B_2= \frac{ NT^2 (R_{max}^2 + A_{max}^2)}{2}$.
\end{lemma}
\begin{proof}
The proof uses the fact that for every $\tau \in [t,t+T-1]$
		\begin{align}
          Q_n(t) - (\tau-t)R_{max} \leq Q_n(\tau) \leq Q_n(t) + (\tau-t)A_{max} \notag
     	\end{align}
and uses these upper and lower bounds for $Q_n(\tau) $ on the R.H.S of \eqref{eq:bound2} and the rest of the proof is	similar to the proof of  Lemma 1 and omitted here.
\end{proof}
Lemma 2 reveals  that now the optimal control actions can be taken by maximizing the third term in the R.H.S of \eqref{eq:bound3},  which yields the optimal solution but with a higher average queue delay. However, we still need  the future channel information (i.e., transmission rates ) to maximize the third term in the R.H.S of~\eqref{eq:bound3} optimally. Here, we exploit one of  key benefits of a Massive MIMO system: under certain condition (i.e., Rayleigh fading channel) the fluctuation over the transmission rates due to small scale fading can be neglected and it only depends on the large-scale fading such as  path-loss and shadowing. This is known as \textit{channel hardening effect}~\cite{7880691} that occurs when the number of antennas at the BS is sufficiency large, which is the case for Massive MIMO systems. Hence, the user rates\footnote{ The large-scale fading can be measured by the BS over a longer time interval.}  become nearly deterministic and simplifies the scheduling and resource allocation problem. This channel characteristic and the result in Lemma 2 reduce  Opt 1 to the following simple optimization problem:

\textbf{Opt 2:}
\begin{equation*}
\max_{I_n^s(t),I_n(t)}  \left[ \sum_{n=1}^{N} Q_n(t) R_n(t) -\frac{CVI^s(t)}{T}  \right]
\label{eq:obj}
\end{equation*}
In Opt 2, $Q_n(\tau)$ and  $R_n(\tau)$  are replaced by $Q_n(t)$ and $R_n(t)$ due to the Lemma 2 and the hardening effect, respectively, and Opt 2 becomes a deterministic problem. We define $\mathcal{N}^s(t)$ as the set of users with SRS at the beginning of time $t$ and derive the following algorithm  that solves Opt 2 optimally. 

\textbf{Joint Scheduling and SRS Allocation (JSSA):}
\begin{itemize}
	\item Input: $V$, $C$, $T$, $P$, $\mathcal{N}^s(t)$. 
	\item Step 1.1 (\textbf{SRS Allocation}): At every $t=lT$, \textit{among all users},  for each set with size $k$, $1  \leq k \leq K$ do:
	\begin{itemize}
	 \item Set the transmit power to $p_n(t)=\frac{P_{tot}}{k}$ 
	 \item Find 
			\begin{align*}
	           \mathcal{S}^*_{k}(t) = \argmax-k \left\{  Q_n(t)R_n(t) \right\}, n \in  \mathcal{N}
	        \end{align*}
	        where $ \argmax - k $ choose the first $k$ elements of a  given set of numbers sorted
	        in decreasing order. $\mathcal{S}^*_{k}(t)$ is called as the best set with size $k$. 
	        \item Find the weight  $W^*_k(t) = \sum_{n \in\mathcal{S}^*_{k}} Q_n(t)R_n(t)$
	  \end{itemize}
   \item Step 1.2: Find 
   			\begin{align*}
             k^*(t) = \argmax_{1 \leq k \leq K} W^*_k(t)
            \end{align*}
    \item Step 1.3: Determine the set  $\mathcal{S}^*_{k^*}(t) $    and  $W^*_{k^*} = \sum_{n \in \mathcal{S}^*_{k^*}(t)} Q_n(t)R_n(t)$. Set $\mathcal{S}_1(t) = \mathcal{S}^*_{k^*}(t)$ and the maximum weight $W_1(t) =W^*_{k^*}(t)  $.
    \vspace{2mm}
    \item Step 1.4: Repeat Step 1.1, Step 1.2 and Step 1.3   \textit{only for the users} $n \in  \mathcal{N}^s(t)$ and determine the best and the maximum weight denoted as  $\mathcal{S}_2(t)$ and $W_2(t)$ as Step 1.2 and Step 1.3, respectively.
    
        \item Step 1.5: if     $W_1(t) - \frac{CV}{T} > W_2(t)$, then reconfigure SRS allocation  and assign SRS to the users in $\mathcal{S}_1(t)$ and update $\mathcal{N}^s(t) $.  Otherwise, do not change the SRS configuration.
              
        \item Step 2 (\textbf{Scheduling}): If the condition in Step 1.5 is satisfied, at every $\tau$,  $\tau \in [t, t+T-1]$ schedule the users in $\mathcal{S}_1(t)$, and otherwise, schedule the users in $\mathcal{S}_2(t)$.          
\end{itemize}
JSSA decides on how many and which users must be chosen among all users and determines that it is worth to reconfigure SRS by comparing the performance achieved among the users that have already SRS. As JSSA optimally minimizes the R.H.S of (9) with the hardening effect,  we have the following Theorem that shows performance of JSSA.
\begin{theorem}(Lyapunov Optimization)
	\label{thm:1}
	Suppose $\boldsymbol \lambda$ is an interior point in $\Lambda$, and there
	exits $\boldsymbol \epsilon > 0$ such that $\boldsymbol \lambda + \boldsymbol \epsilon \in
	\Lambda$. Then, under JSSA, we have
	the following bounds:
		\begin{align}
	\limsup_{L\rightarrow \infty} \frac{1}{L} \sum_{l=0}^{L-1}
	\mathbb{E}[C I^s(lT)] \leq C_{avg}^* +  \frac{B_2}{V}\\
		\limsup_{L\rightarrow \infty} \frac{1}{L} \sum_{l=0}^{L-1}
	\sum_{n=1}^{N} \mathbb{E}[Q_n(lT)] \leq  \frac{B_2 + VC}{\epsilon}
	\label{eq:thr1}
	\end{align}
	Where $C_{avg}^*$ is the optimal solution of problem (2)-(3)-(4).
\end{theorem}
\begin{proof}
To avoid redundancy with existing literature, we omit the details here. The sketch of the proof is as follows: it follows similar steps in Theorem 5.4 of~\cite{Georgiadis:Resource06} by first showing the existence of  a stationary randomized algorithm   that is optimal and achieves the minimum time average cost by choosing the control actions independently from $\mathbf{Q}(t)$ but according to a fixed probability distribution known to the system. Then, it is shown that JSSA is better than the randomized algorithm in minimizing the R.H.S of (9) and thus it is also optimal. 
\end{proof}
Theorem 1 implies that the average cost under JSSA approaches to  the optimal cost $C_{avg}^*$ as $V$  increase, while the average  queue sizes  also increase.
\vspace{-4 mm}
\section{Numerical Analysis}
\label{sec:sim} In our simulations, there  is a single cell covering a square of 250 m x 250 m area with a Massive MIMO capable BS. We set $M=64$ and there are $N=300$ users uniformly and independently distributed in
the cell. We adapt the same channel models and take the related parameters given in~\cite{massivemimobook} for large-scale  and  NLOS fading. We apply MMSE precoding, and set $B= 20$ MHz, $K=10$ and $P_{tot} =1$ Watts.  We assume that $\tau = 1$ millisecond, and at each time slot $P=60$ users send their SRS to the BS,  and $\gamma=0.8$. 3GPP FTP Model 3 is considered,  where user traffic follows Poisson arrival process  with a  payload size 0f 0.2 MB and different mean arrival rate varying between 0.5 and 2 seconds.

%We first show the effect of dynamic SRS configuration and allocation on the network throughput. In Figure~\ref{fig:Fig1}, the total average network throughput under a Random and M-JSSA algorithms with different number of users is depicted. Random algorithm assigns SRS to the users randomly once and the allocation does not change throughout the simulation. M-OJSSA is modified version of JSSA where the cost associated to SRS configuration is not considered, and hence at every $T=20$ ms seconds M-JSSA reconfigures SRS allocation. From the Figure when  $N=40$,  every user receives SRS since $P=60$, and the performance of Random and MSSA is same. As $N$ increase, the performance of Random algorithm decreases significantly since Random algorithm cannot keep track and assign SRS to the users with data. The important result is that Massive and MU-MIMO capability cannot be utilized without proper SRS configuration.
%\begin{figure}
%	\centering
	%\begin{left}
	%     \includegraphics[width=0.5\columnwidth]{fig_regions.eps}
%	\includegraphics[width=0.9\columnwidth]{Fig1.eps}
	%\end{left}
%	\caption{Total avg. throughput with different number of users. }
%	\label{fig:Fig1}
%\end{figure}
\begin{figure}
	\centering
	%\begin{left}
	%     \includegraphics[width=0.5\columnwidth]{fig_regions.eps}
	\includegraphics[width=0.8\columnwidth]{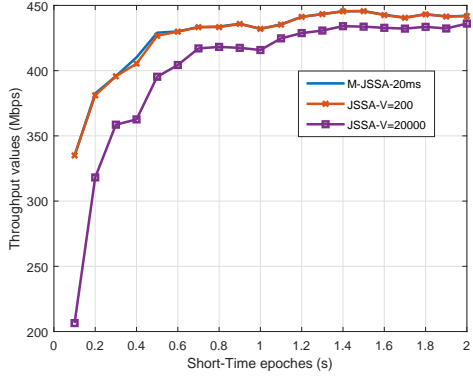}
	%\end{left}
	\caption{Time avg. total throughput with different $V$. }
	\label{fig:Fig2}
\end{figure}
\begin{figure}[th!]
	\centering
	%\begin{left}
	%     \includegraphics[width=0.5\columnwidth]{fig_regions.eps}
	\includegraphics[width=0.8\columnwidth]{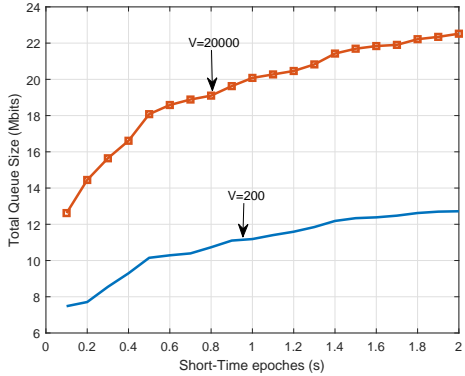}
	%\end{left}
	\caption{Time  avg. total queue size with different $V$. }
	\label{fig:Fig3}
\end{figure}
We first show the performance of JSSA when  $T=20$ ms. Figure~\ref{fig:Fig2} depicts the time average (every 100 ms) total network throughput achieved by JSSA with different values of $V$.  Modified JSSA (M-JSSA) configures SRS without any cost at every $T$ and it constitutes a benchmark to the performance of JSSA. When $V=200$, JSSA achieves almost the same throughput as that of M-JSSA and the average number of SRS configuration is approximately \textbf{ 0.94}, which implies that more than 90\% of the time JSSA attemps to reconfigure SRS and hence the  cost is high. We observe that when $V=200$ the average total queue size approaches to a fixed point as shown in Figure~\ref{fig:Fig3}, which means the network is stabilized . 

As $V$ increases to $V=20000$, the throughput achieved by JSSA and the benachmark approach to the same value,  which reveals  that the network is still stable. However, the average number of SRS reconfiguration is reduced to \textbf{0.3}, so  the cost decreases. We also observe from Figure~\ref{fig:Fig3} that the average total queue size is higher when $V=20000$ compared to the case with $V=200$, that is aligned with the theoretical result found in Theorem 1.
\vspace{-4 mm}
\section{Conclusion}
\label{sec:conclusion}  We have investigated the problem of SRS allocation and scheduling problem in a single cell Massive MIMO network. By applying Lyapunov optimization tool, we have developed a joint scheduling and SRS allocation algorithm that can perform well under random traffic and channel conditions. In simulation results, we show that the average signaling cost can be reduced at the expense of an increase in the average queue delay. The problems of the SRS allocation with different objectives (i.e., reducing delay)  with fairness can be other research directions. Also, SRS allocation in a multi-cell setup with pilot contamination would be an interesting research problem.

%%%%%%%%%%%%%%%%%%%%%%%%%%%%%%%%%%%%%%%%%%%%%%%%%%%%%%%%%%%%%%%%%%%%%%%%%%%%%%
%%%%%%%%%%%%%%%%%%%%%%%%%%%%%%%%%%%%%%%%%%%%%%%%%%%%%%%%%%%%%%%%%%%%%%%%%%%%%%
\vspace{-4 mm}
\appendices
\section{Proof of Lemma \ref{lemma:1}}
\label{sec:lemma1}
The proof starts with finding an upper bound for the Lyapunov drift given in \eqref{eq:drift} by using the following fact: for user $n$, the following inequality holds.
\begin{align}
Q_n^2(\tau+1) -  & Q_n^2(\tau)  \leq   R_n^2(\tau) + A_n^2(\tau)\\ \notag
& - 2Q_n(\tau) \left [R_n(\tau) -  A_n(\tau)\right]
\end{align}
By summing (12) over   $[t,t+T-1]$ and knowing that $R_n(\tau) \leq R_{max}$ and $A_n(\tau) \leq A_{max}$  for all $n$ and $\tau$  we obtain,
\begin{align}
Q_n^2(t+T) -  & Q_n^2(t)   \leq   TR^2_{max} + TA^2_{max} \\\notag
& - 2 \left [ \sum_{\tau = t}^{t +T-1}Q_n(\tau)[R_n(\tau) - A_n(\tau)]\right]
\end{align}
Then,  by taking  the conditional expectation of (13) with respect to  $\mathbf{Q(t)}$ and summing over all users, and dividing by  1/2 we have,
\begin{align*}
\Delta_T(t)  \leq   B_1  - \mathbb{E} \left[\sum_{\tau = t}^{t +T-1} \sum_{n=1}^{N} Q_n(\tau)[R_n(\tau) - A_n(\tau) ]   | \mathbf{Q(t)}\right] 
\end{align*}
Finally, we add the penalty term $V \mathbb{E} \left[ I^s(t) | (\mathbf{Q(t)}\right] $ to both sides of  above inequality and rearranging the resulting terms, we have the bound in Lemma 1.
%\begin{align*}
%\Delta_T(t) & + V \mathbb{E} \left[ I^s(t) | (\mathbf{Q(t)}\right]  \nonumber\\
%&\leq   B_1  +\mathbb{E} \left[\sum_{\tau = t}^{t +T-1} \sum_{n=1}^{N} Q_n(\tau)A_n(\tau)   | \mathbf{Q(t)}\right] %\nonumber\\
%& - \mathbb{E} \left[\sum_{\tau = t}^{t +T-1} \sum_{n=1}^{N} Q_n(\tau) R_n(\tau) -CVI^s(t)   | \mathbf{Q(t)} \right]\
%\end{align*}
This completes the proof.
%\newpage
\vspace{-3 mm}
\bibliographystyle{IEEEtran}
\bibliography{IEEEabrv,ref}

%%%%%%%%%%%%%%%%%%%%%%%%%%%%%%%%%%%%%%%%%%%%%%%%%%%%%%%%%%%%%%%%%%%%%%%%%%%%%
%%%%%%%%%%%%%%%%%%%%%%%%%%%%%%%%%%%%%%%%%%%%%%%%%%%%%%%%%%%%%%%%%%%%%%%%%%%%%

\end{document}